\documentclass[conference,a4paper]{IEEEtran}
\usepackage{amsmath,amssymb,mathrsfs}
\usepackage{epsfig,epsf,subfigure,graphicx,graphics}
\usepackage{url}

\newtheorem{theorem}{Theorem}[section]

\newtheorem{remark}{Remark}[section]

\newtheorem{fact}{Fact}[section]

\newcommand{\E}{\mathrm{E}}

\newcommand{\sym}{\mathrm{sym}}

\newcommand{\eqFunc}{\overset{\mathrm{f}}{=}}

\newcommand{\aaaa}{\mathrm{(a)}}
\newcommand{\bbbb}{\mathrm{(b)}}

\newcommand{\lp}{\left(}
\newcommand{\rp}{\right)}
\newcommand{\lb}{\left[}
\newcommand{\rb}{\right]}
\newcommand{\lbp}{\left\{}
\newcommand{\rbp}{\right\}}
\newcommand{\ul}{\underline}
\newcommand{\ol}{\overline}
\newcommand{\mcal}{\mathcal}

\newcommand{\wtild}{\widetilde}
\newcommand{\mb}{\mathbf}
\newcommand{\mbb}{\mathbb}
\newcommand{\msf}{\mathsf}

\newcommand{\ra}{\rightarrow}

\newcounter{MYtempeqncnt}

\graphicspath{{./figs/}}

\IEEEoverridecommandlockouts

\allowdisplaybreaks

\title{Interference Channel with Intermittent Feedback}
\author{
\authorblockN{Can Karakus}
\authorblockA{
UCLA, Los Angeles, USA\\
\textsf{karakus@ucla.edu}}
\and
\authorblockN{I-Hsiang Wang}
\authorblockA{
EPFL, Lausanne, Switzerland\\
\textsf{i-hsiang.wang@epfl.ch}}
\and
\authorblockN{Suhas Diggavi}
\authorblockA{
UCLA, Los Angeles, USA\\
\textsf{suhasdiggavi@ucla.edu}}
}

\begin{document}
\maketitle
\begin{abstract}
  We investigate how to exploit intermittent feedback for interference
  management. Focusing on the two-user linear deterministic
  interference channel, we completely characterize the capacity
  region. We find that the characterization only depends on the
  forward channel parameters and the marginal probability distribution
  of each feedback link. The scheme we propose makes use of block
  Markov encoding and quantize-map-and-forward at the transmitters,
  and backward decoding at the receivers. Matching outer bounds are
  derived based on novel genie-aided techniques. As a consequence, the
  perfect-feedback capacity can be achieved once the two feedback
  links are active with large enough probabilities.
\end{abstract}

\section{Introduction}
The simplest information theoretic model for studying interference is
the two-user Gaussian \emph{interference channel} (IC). It is shown
that feedback can provide an unbounded gain in capacity for two-user
Gaussian interference channels \cite{SuhTse_11}, in contrast to the
bounded power gain provided by feedback in point-to-point, multiple
access, and broadcast channels. This has been demonstrated when the
feedback is unlimited, perfect, and free of cost in \cite{SuhTse_11}.
This motivates the natural question of whether feedback can provide
similar gains under more practical feedback models.

In this work, we investigate how to exploit \emph{intermittent}
feedback for managing interference.  Such intermittent feedback could
occur in several situations. For example, one could use a side-channel
such as  WiFi for feedback; in this case since the WiFi channel is
best effort, dropped packets might cause intermittent feedback.  In
other situations, control mechanisms in higher network layers could
cause the feedback resource to be available intermittently. We study
the effect of intermittent feedback using the linear deterministic
model \cite{AvestimehrDiggavi_09} of the two-user Gaussian IC. For the
feedback links, Bernoulli processes $\{S_1[t]\}$ and $\{S_2[t]\}$
control the presence of feedback for user $1$ and $2$,
respectively. Although the joint distribution $p (S_1 [t], S_2 [t])$ can be
time-variant in general, for simplicity, we focus on the case where it is i.i.d. over time.
Our results suggest that extension to the time-variant case is straightforward.
We assume that the receivers are
\emph{passive}: they simply feedback their received signals back to
the transmitters without coding. In other words, each transmitter
receives from feedback a punctured version of the received sequence at
its own receiver with unit delay. We focus on the passive feedback model
as the intermittence of feedback is motivated by
the availability of feedback resources (either through use of best-effort
WiFi for feedback or through feedback resource scheduling). Therefore, 
it might be that the time-variant statistics of the intermittent feedback are not \emph{a priori}
available at the receiver and therefore precluding active coding. Moreover, the
availability of the feedback resource may not be known ahead of 
transmission, therefore motivating the causal state-information at the transmitter.

In the literature, other practical feedback models are also
investigated.  Rate-limited feedback for the two-user IC was
considered in \cite{VahidSuh_12}, where the feedback from the
receivers to the transmitters is modeled by two finite-capacity
noiseless links. \cite{VahidSuh_12} characterized the capacity region
for the linear deterministic IC and the sum capacity to within a
constant gap for the symmetric Gaussian IC. If indeed the feedback 
statistics is \emph{a priori} known at the receiver, one can use active feedback to
code for these erasures thereby creating noiseless finite capacity feedback 
links. In contrast to our model,
the receivers in \cite{VahidSuh_12} can actively code the feedback
signals. On the other hand, \cite{LeTandon_12} considered feedback
with additive white Gaussian noise but the receivers are passive, that
is, they cannot encode the feedback signal. The capacity region is
characterized for the symmetric linear deterministic IC
\cite{LeTandon_12}. When feedback shares the same resource (spectrum)
with the forward channel and hence is not free of cost, it is shown in
\cite{SuhWang_12} through the study on two-way interference channels
that feedback can provide net capacity gain even taking the feedback
cost into account.

Our main contribution is the characterization of the capacity region
of the linear deterministic IC with intermittent feedback, which only
depends on the forward channel parameters and the marginal
distribution of $S_1$ and $S_2$, not on the joint distribution.
Interestingly, the full benefit obtained via perfect feedback can be
achieved once the ``on" probabilities of the two feedback links are
large enough.
Our result in the linear deterministic model also suggests that in the
Gaussian case, the capacity gain from intermittent feedback remains
unbounded.

We propose a block Markov encoding scheme to exploit intermittent
feedback at the transmitters along with backward decoding at the
receivers. For linear deterministic IC, the high-level idea is to
exploit the additional information provided by intermittent
feedback to refine the interfered signals or to relay additional
information.  Due to the passive nature of the receivers, not all the
information contained in the feedback is useful, in sharp contrast to
the case of rate-limited feedback \cite{VahidSuh_12}.  Therefore, at
the transmitters instead of the (partial) decode-and-forward scheme
employed in \cite{SuhTse_11} and \cite{VahidSuh_12}, we use
\emph{quantize-map-and-forward} \cite{AvestimehrDiggavi_09} 
to extract useful information from the intermittent feedback and 
send it to the receivers.

We also develop novel outer bounds that match the achievable rate
region. Remarkably, in our proof we do not make use of the assumption
that the states are known to the transmitters causally, which proves
that even when the realization of $(S_1^N, S_2^N)$ is known
non-causally, the capacity region remains the same.


The rest of this paper is organized as follows. We formulate the
problem in Section~\ref{sec_Formulation} and present our main result
in Section~\ref{sec_Result}. Achievability and converse are proved in
Section~\ref{sec_Achieve} and Section~\ref{sec_Converse} respectively.
An extended version of this paper that includes the proof details can be found in \cite{KarakusWang_13}.

\section{Problem Formulation}\label{sec_Formulation}
\begin{figure}[htbp]
{\center
\includegraphics[width=2in]{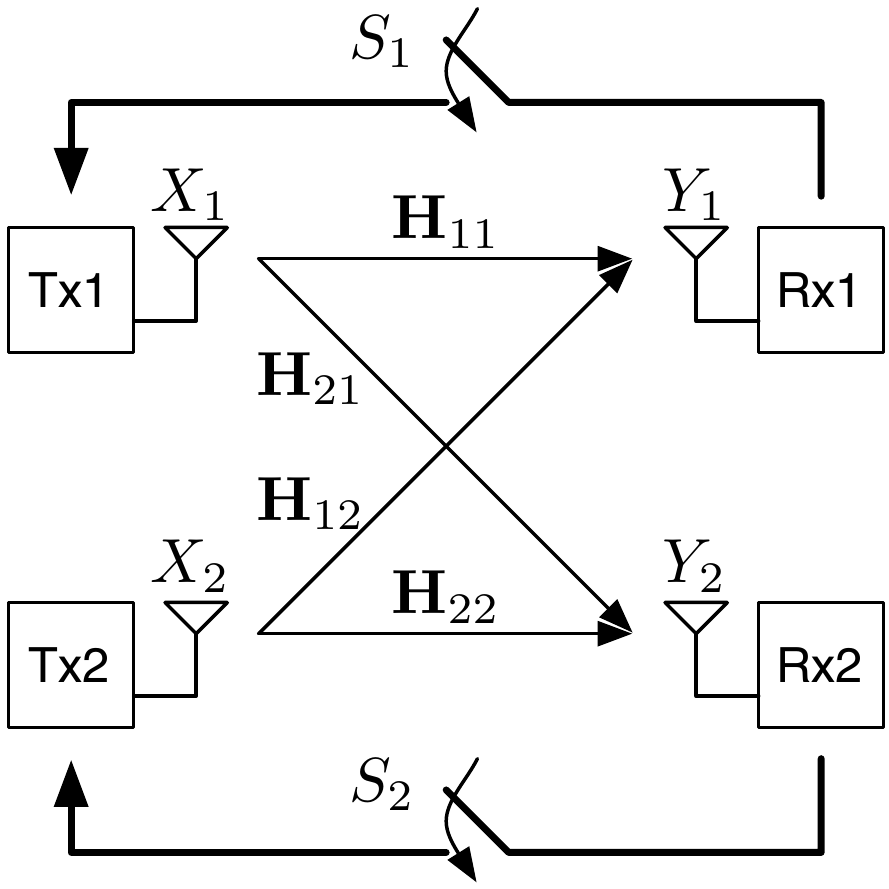}
\caption{Linear Deterministic IC with Intermittent Feedback}
\label{fig_Channel}
}
\end{figure}

In this paper, we focus on the linear deterministic model \cite{AvestimehrDiggavi_09} of the two-user Gaussian IC. An illustration is given in Fig.~\ref{fig_Channel}.

The transmitted signal at transmitter $i$ (Tx$i$) is $X_i \in \mbb{F}_2^{q}$, for $i=1,2$. Here $\mbb{F}_2$ denotes the binary field $\{0,1\}$. The received signals at receiver $1$ (Rx1) and receiver $2$ (Rx2) are
\begin{align*}
Y_1[t] &= \mb{H}_{11} X_1[t] + \mb{H}_{12} X_2[t], \\ 
Y_2[t] &= \mb{H}_{22} X_2[t] + \mb{H}_{21} X_1[t],
\end{align*}
where additions are modulo-two component-wise. Channel transfer matrices $\mb{H}_{ij} := \mb{S}^{q-n_{ij}}$ for $(i,j) \in \{1,2\}^2$, where $q = \max\lbp n_{11},n_{12},n_{21},n_{22}\rbp$, and $\mb{S}\in\mathbb{F}_2^{q\times q}$ is the shift matrix
$
\begin{bmatrix}
\mb{0}^T & 0\\ \mb{I}_{q-1} & \mb{0}
\end{bmatrix}
$, where $\mb{0}$ is the zero vector in $\mbb{F}_2^{q-1}$ and $\mb{I}_{q-1}$ is the identity matrix in $\mbb{F}_2^{(q-1)\times(q-1)}$.
The transmit signal from Tx$i$ at time $t$, $X_i[t]$, is determined by the message $W_i$, $\wtild{Y}_i[1:t-1]$, and $(S_1[1:t-1],S_2[1:t-1])$, where $\wtild{Y}_i[t] := S_i[t]Y_i[t]$. 
Let us use the notation $A\eqFunc B$ to denote that $A$ is a function of $B$. Then,
$X_i[t] \eqFunc \lp W_i, \wtild{Y}_i[1:t-1], S_1[1:t-1],S_2[1:t-1]\rp$.

The feedback state sequences have the joint distribution
$$ p\lp S_1^N, S_2^N\rp = \prod_{t=1}^Np\lp S_1[t],S_2[t]\rp.$$
Let $q_{i_1i_2} := p\lp S_1=i_1, S_2=i_2\rp$, $i_1,i_2\in\{0,1\}$. 
Marginally, $S_i[t]\sim\mathrm{Ber}(p_i)$, i.i.d. over time. In other words, the feedback signal is erased with probability $1-p_i$, where 
$$p_1 := q_{10}+q_{11},\ p_2 := q_{01}+q_{11}.$$

For notational convenience, denote $\ul{S} := (S_1,S_2)$ and 
\begin{align*}
&V_1:=\mb{H}_{21} X_1,&
&V_2:=\mb{H}_{12} X_2,&
&\wtild{V}_1 := S_2V_1, & &\wtild{V}_2 := S_1V_2.
\end{align*}


\section{Main Result}\label{sec_Result}
\begin{figure*}[!t]
\normalsize
\setcounter{MYtempeqncnt}{\value{equation}}
\setcounter{equation}{0}
\begin{align}
R_1 &\le \min\lbp \max(n_{11},n_{12}), n_{11}+p_2(n_{21}-n_{11})^+\rbp \label{eq_R1Bd}\\
R_2 &\le \min\lbp \max(n_{22},n_{21}), n_{22}+p_1(n_{12}-n_{22})^+\rbp \label{eq_R2Bd}\\
R_1+R_2 &\le \min\Big\{ \max(n_{11},n_{12}) + (n_{22}-n_{12})^+ , \max(n_{22},n_{21}) + (n_{11}-n_{21})^+ \Big\}\label{eq_R1R2Bd1}\\
R_1+R_2 &\le \max\lbp n_{12}, (n_{11}-n_{21})^+\rbp + \max\lbp n_{21}, (n_{22}-n_{12})^+\rbp \notag\\
&\quad + p_1\min\lbp n_{12}, (n_{11}-n_{21})^+\rbp + p_2\min\lbp n_{21}, (n_{22}-n_{12})^+\rbp \label{eq_R1R2Bd3}\\
2R_1+R_2 &\le \max(n_{11},n_{12}) + \max\lbp n_{21}, (n_{22}-n_{12})^+\rbp + (n_{11}-n_{21})^+ + p_2\min\lbp n_{21}, (n_{22}-n_{12})^+\rbp \label{eq_2R1R2Bd}\\
R_1+2R_2 &\le \max(n_{22},n_{21}) + \max\lbp n_{12}, (n_{11}-n_{21})^+\rbp + (n_{22}-n_{12})^+ + p_1\min\lbp n_{12}, (n_{11}-n_{21})^+\rbp \label{eq_R12R2Bd}
\end{align}
\addtocounter{MYtempeqncnt}{6}
\setcounter{equation}{\value{MYtempeqncnt}}
\hrulefill
\end{figure*}

The main result is summarized in the following theorem.
\begin{theorem}[Capacity Region]\label{thm_Capacity}
The capacity region $\mcal{C}$ for the linear deterministic IC with intermittent feedback is the collection of non-negative $(R_1,R_2)$ satisfying \eqref{eq_R1Bd} -- \eqref{eq_R12R2Bd}.
\end{theorem}

Note that the rate region defined by \eqref{eq_R1Bd} -- \eqref{eq_R1R2Bd1} is the perfect-feedback IC capacity region \cite{SuhTse_11}. Therefore, once $p_1$ and $p_2$ are so large that  \eqref{eq_R1R2Bd3} -- \eqref{eq_R12R2Bd} become inactive, the perfect-feedback performance can be attained even under intermittent feedback. The thresholds on $p_1,p_2$ will depend on $\{n_{ij},i,j\in\{1,2\}\}$.

As an example, let us focus on the symmetric capacity $C_{\sym}$ under the symmetric setting $n_{11}=n_{22}=n, n_{12}=n_{21}=\alpha n, p_1=p_2=p$: ($C_{\sym} := \max_{(R,R)\in\mcal{C}} R$)
\begin{align*}
\frac{C_{\sym}}{n} &= \lbp\begin{array}{ll}
\min\lbp 1-\alpha/2, 1-(1-p)\alpha\rbp, &\alpha \le 1/2\\
\min\lbp 1-\alpha/2, p+(1-p)\alpha\rbp, &1/2 \le \alpha \le 1\\
\min\lbp \alpha/2, (1-p) + p\alpha\rbp, & \alpha \ge 1
\end{array}\right. ,
\end{align*}
where the first term in each minimization is the perfect-feedback capacity. Hence, we find the threshold on $p$ above which perfect-feedback capacity can be achieved, as follows:
\begin{align*}
p^* &= \lbp\begin{array}{ll}
1/2, &\alpha \le 1/2\\
(2-3\alpha)^+/(2-2\alpha), &1/2 \le \alpha \le 1\\
(\alpha-2)^+/(2\alpha-2), & \alpha \ge 1
\end{array}\right. .
\end{align*}
Note that in the regime $2/3\le \alpha\le 2$, feedback does not increase the symmetric capacity of IC \cite{SuhTse_11} and hence $p^*=0$, that is, we do not need feedback at all. Also note that $p^* \le 1/2$ for all $\alpha$. Therefore, once $p \ge 1/2$, perfect-feedback capacity can be achieved with intermittent feedback regardless of channel parameters $\alpha$ and $n$. Note that the larger $p$ is, the larger the amount of additional information about the past reception can be obtained through intermittent feedback at the transmitters. If the amount of such information is larger than a threshold, then sending it to the receivers will limit the rate for delivering fresh information. The threshold $p^*$ represents the maximum limit at which the help of feedback is not neutralized by this effect.

%
%
\section{Achievability Proof}\label{sec_Achieve}
To prove the achievability part of Theorem~\ref{thm_Capacity}, in this
section we provide a coding scheme to exploit intermittent feedback in
the interference channel. Since the feedback is passive, one cannot
code against the erasures in the feedback links. This is the key
difference with the active feedback case \cite{VahidSuh_12}. In the
active feedback case, a block Markov coding scheme based on
decode-and-forward at the transmitters is employed, which is a natural
extension of that in the perfect feedback
case \cite{SuhTse_11}. Instead, we employ a block Markov coding scheme
based on \emph{quantize-map-and-forward}, which can be viewed as a
non-trivial extension of the perfect feedback scheme. Below we
describe the scheme in detail.

\subsection{High-level Description}
For the two-user linear deterministic IC, Han-Kobayashi coding scheme is a natural choice where we split the message into common and private: $W_i := (W_{ic}, W_{ip})$ for $i=1,2$. The total number of blocks to be transmitted is $B$, and the length of each block is $N$. For block $b\in[1:B]$, the messages $\{W_1(b), W_2(b)\}$ are independent from the messages of other blocks. 
In the following we describe the encoding and decoding for a particular block $b$ for user $1$. Operations of user $2$ are similar. 

In the beginning of block $b$, from the punctured feedback $\wtild{Y}_1^N(b-1)$ and the state sequence $\ul{S}^N(b-1)$, Tx1 generates
\begin{align*}
&\wtild{V}_1^N(b-1) = S_2^N(b-1)V_1^N(b-1),\\
&\wtild{V}_2^N(b-1) = \wtild{Y}_1^N(b-1) - S_1^N(b-1)X_1^N(b-1).
\end{align*}

Note that Tx2 can also generate $\lp \wtild{V}_1^N(b-1), \wtild{V}_2^N(b-1) \rp$. The high-level idea is to use this common information to \emph{cooperatively} refine the previously received signals at the receivers and/or relay additional information, as in the perfect feedback case \cite{SuhTse_11}. The only difference is that here we do not decode at the transmitters. The reason is that the imperfect feedback may lay additional constraints on the achievable rate if we insist to decode. Decoding partially is not optimal since the realization of erasures in the feedback links is not known beforehand, and if we insist to decode at the transmitters some pre-assigned sub-messages, this could harm the achievable rate of other sub-messages.

Hence, we shall quantize $\wtild{V}_1^N(b-1)$ and $\wtild{V}_2^N(b-1)$, and map the quantized outputs to a codeword $X_{1e}^N(b)$ which contains the helping information for interference refinement or message relaying regarding block $b-1$. On top of it, we further superpose fresh information of block $b$, the messages $W_{1c}(b)$ and $W_{1p}(b)$, to generate the transmit codeword $X_1^N(b)$. 

For decoding, we employ backward decoding. At the end of block $b$, assuming that quantized version of $\wtild{V}_1^N(b)$ and $\wtild{V}_2^N(b)$ has been successfully decoded from the future block $b+1$, Rx1 decodes $W_{1c}(b)$, $W_{2c}(b)$, and $W_{1p}(b)$ jointly with the quantized version of $\wtild{V}_1^N(b-1)$ and $\wtild{V}_2^N(b-1)$.

In the above scheme, we can see that $\wtild{V}_1^N(b-1)$ and $\wtild{V}_2^N(b-1)$ will contain some information of block $b-2$, which contains information about block $b-3$, and so on. This dependency across blocks hampers a single-letter characterization of the achievable rates, since the mutual information terms obtained will depend on all the signals sent in previous blocks. In order to remove the dependency across more than two blocks, we carry out the following operation before quantizing $\wtild{V}_1^N(b-1)$ and $\wtild{V}_2^N(b-1)$: generate
\begin{align*}
\ol{V}_i^N(b-1) := \wtild{V}_i^N(b-1) - \wtild{V}_{ie}^N(b-1)
\end{align*}
where $\wtild{V}_{ie}^N(b-1) := S_j^N(b-1)\mb{H}_{ji}X_{ie}^N(b-1)$ for $(i,j) = (1,2), (2,1)$. The operation is feasible since both $X_{1e}^N(b-1)$ and $X_{2e}^N(b-1)$ are generated from shared information $\ol{V}_1^N(b-2)$ and $\ol{V}_2^N(b-2)$, and hence (by induction) available at both transmitters. In our proposed encoding architecture, the quantization is performed on $\ol{V}_1^N(b-1)$ and $\ol{V}_2^N(b-1)$. An illustration of the encoding architecture is given in Fig.~\ref{fig_ENC}.

\begin{figure}[htbp]
{\center
\includegraphics[width=3.4in]{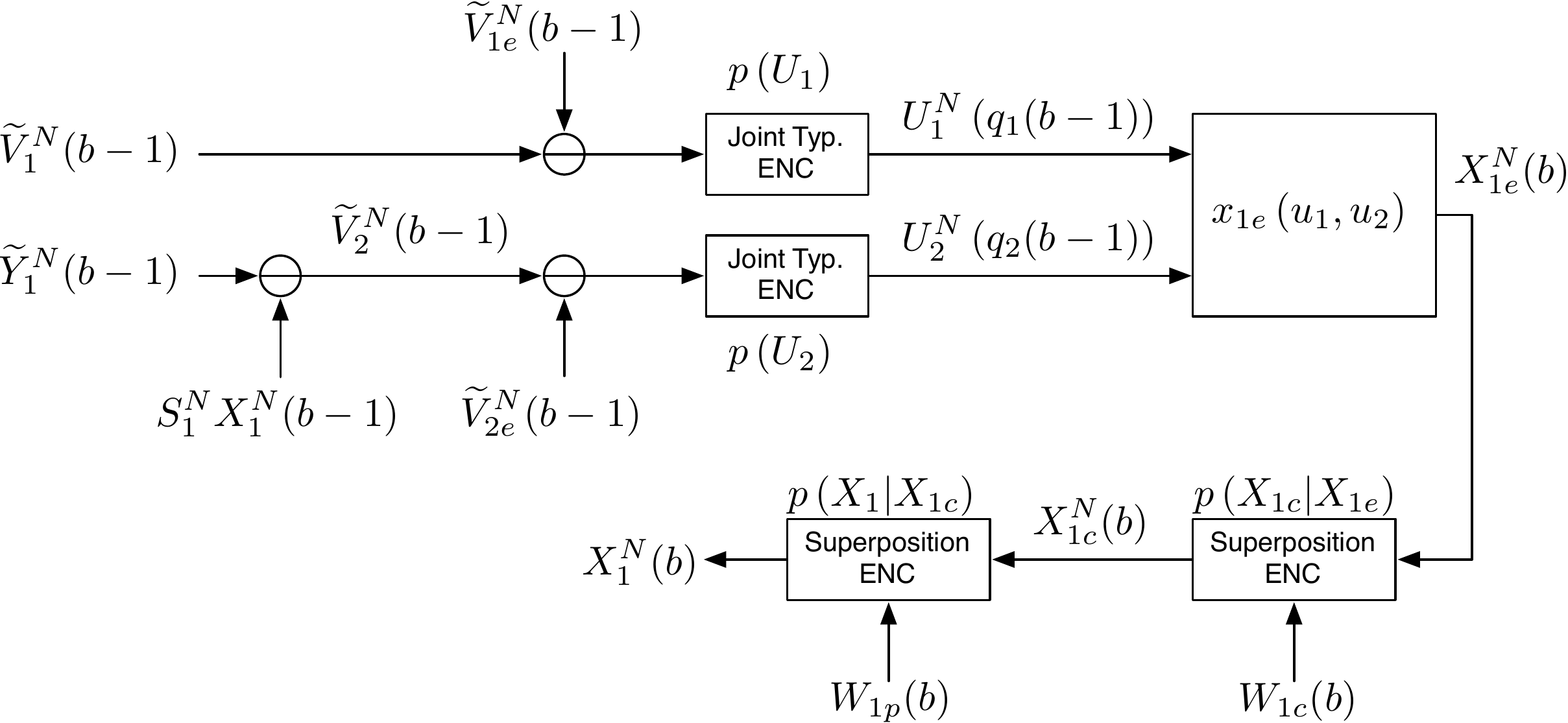}
\caption{Block Diagram of Encoder at Tx1}
\label{fig_ENC}
}
\end{figure}

\subsection{Codebook Generation and Detailed Coding Process}
We describe the scheme for block $b$ in detail below.

{\flushleft\bf Codebook Generation and Encoding}:\par
Based on $p\lp U_i\rp$, generate $2^{Nr_i}$ quantization codewords $U_i^N$ i.i.d. over time, for $i=1,2$, to quantize $\ol{V}_i^N(b-1)$. Let $q_1(b-1)$ and $q_2(b-1)$ denote the quantization indices. The quantization is carried out by a joint typicality encoder as in standard source coding. We then choose a symbol-by-symbol map $x_{1e}(u_1,u_2)$ for mapping $\lp U_1^N(q_1(b-1)), U_2^N(q_2(b-1))\rp$ into $X_{1e}^N(b)$. This completes the joint-source-channel coding part of the encoder. Note that we shall use the same pair of quantization codebooks at both transmitters.

Superposition encoding is done in a standard way. We base on $p\lp X_{1c} | X_{1e}\rp$ to generate $2^{NR_{1c}}$ common codewords $X_{1c}^N( W_{1c}(b), X_{1e}^N(b))$ i.i.d. over time. Then based on $p\lp X_{1} | X_{1c}\rp$, we generate $2^{NR_{1p}}$ transmit codewords $X_{1}^N( W_{1p}(b), X_{1c}^N(b))$ i.i.d. over time.

{\flushleft\bf Decoding}:\par
At the end of block $b$, we assume that quantization indices $\{q_1(b),q_2(b)\}$ have been decoded from block $b+1$. The additional information carried by $\{U_1^N(q_1(b)), U_2^N(q_2(b))\}$ can be used in conjunction with $Y_1^N(b)$ to decode $\lp W_{1c}(b), W_{1p}(b), q_1(b-1), q_2(b-1)\rp = (i_1,j,q_1,q_2)$. We find such a unique $(i_1,j,q_1,q_2)$ and some $W_{2c}(b)=i_2$ such that
\begin{align}\label{eq_Seq}
\lp \begin{array}{l}
Y_1^N(b), U_1^N\lp {q}_1(b)\rp, U_2^N\lp {q}_2(b)\rp,U_1^N(q_1), U_2^N\lp q_2\rp,\\ X_{1c}^N(i_1, q_1,q_2), X_{1}^N(j, i_1,q_1,q_2), X_{2c}^N(i_2, q_1,q_2)
\end{array}\rp
\end{align}
is jointly $\epsilon$-typical.

\subsection{Analysis}
The key to a single-letter rate characterization of the above scheme is that, the actual quantization codewords are independent across different blocks, that is, $\lp U_1^N\lp {q}_1(b)\rp, U_2^N\lp {q}_2(b)\rp \rp$ and $\lp U_1^N\lp {q}_1(b-1)\rp, U_2^N\lp {q}_2(b-1)\rp\rp$ are independent. This is due to the removal of $\wtild{V}_{1e}^N(b-1)$ and $\wtild{V}_{2e}^N(b-1)$ described above. The error probability analysis is standard so we omit the details here. Below we sketch the analysis of the error event where $j$ is decoded incorrectly but $(i_1,i_2,q_1,q_2)$ is correct. Note that the joint distribution of the random vectors in \eqref{eq_Seq} is
\begin{align*}
&p\lp y_1^N, u_1^N, u_2^N | u_1'^N, u_2'^N, x_{1}^N, x_{2c}^N\rp p\lp u_1'^N, u_2'^N\rp\\
&\cdot p\lp x_1^N | x_{1c}^N \rp p\lp x_{1c}^N| u_1'^N, u_2'^N\rp p\lp x_{2c}^N| u_1'^N, u_2'^N\rp
\end{align*}
with a change of notations for the sake of simplicity:
\begin{align*}
&Y_1^N(b)\ra y_1^N,\ U_1^N\lp {q}_1(b)\rp\ra u_1^N,\ U_2^N\lp {q}_2(b)\rp\ra u_2^N,\\
&U_1^N(q_1)\ra u_1'^N,\ U_2^N\lp q_2\rp\ra u_2'^N,\ X_{1c}^N(i_1, q_1,q_2)\ra x_{1c}^N,\\
&X_{1}^N(j, i_1,q_1,q_2)\ra x_1^N,\ X_{2c}^N(i_2, q_1,q_2)\ra x_{2c}^N.
\end{align*}
By packing lemma \cite{El-GamalKim_11}, the probability of this error event vanishes as $N\ra\infty$ if 
\begin{align*}
R_{1p} &\le I\lp X_1;Y_1, U_1, U_2 | X_{1c},X_{2c},U_1', U_2'\rp.
\end{align*}

Analysis of the other error events follows similarly. For the joint-typicality encoding to be successful with high probability, we need $r_i \ge I\lp U_i; \ol{V}_i\rp$ for $i=1,2$, due to covering lemma \cite{El-GamalKim_11}. 
Hence, decoding is guaranteed to be successful with high probability if the following holds: for $(i,j)=(1,2)$ and $(2,1)$,
\begin{align*}
R_{ip} &\le I\lp X_i;Y_i, U_i, U_j | X_{ic},X_{jc},U_i', U_j'\rp\\
R_{jc} + R_{ip} &\le I\lp X_{jc}, X_i;Y_i, U_i, U_j | X_{ic},U_i', U_j'\rp\\
R_i &\le I\lp X_i;Y_i, U_i, U_j | X_{jc},U_i', U_j'\rp\\
R_{jc} + R_i &\le I\lp X_{jc}, X_i;Y_i, U_i, U_j |U_i', U_j'\rp\\
r_i + r_j + R_{jc} + R_i &\le I\lp U_i', U_j', X_{jc}, X_i;Y_i, U_i, U_j \rp\\
r_i &\ge I\lp U_i; \ol{V}_i\rp.
\end{align*}
for some mapping functions $\{ x_{1e}(u_1,u_2), x_{2e}(u_1,u_2)\}$ and input distribution
\begin{align*}
&p\lp U_1, U_2\rp p\lp U_1', U_2'\rp p\lp X_{1c}|U_1',U_2'\rp p\lp X_1|X_{1c}\rp\\
&\cdot p\lp X_{2c}|U_1',U_2'\rp p\lp X_2|X_{2c}\rp.
\end{align*}
Here $R_i=R_{ip}+R_{ic}$ is the achievable rate for user $i$, $i=1,2$. $\ol{V}_i := \wtild{V}_i - \wtild{V}_{ie}$ as defined previously. $\lp U_1', U_2'\rp$ corresponds to the $(U_1^N(b-1), U_2^N(b-1))$ while $\lp U_1, U_2\rp$ corresponds to the $(U_1^N(b), U_2^N(b))$. Hence, $p\lp U_1', U_2'\rp$ and $p\lp U_1, U_2\rp$ should be the same, since we use the same distribution to generate the quantization codebooks in all the blocks.

\subsection{Rate Region Evaluation}
To achieve the capacity region, we choose the input distribution and the mapping function as follows. For the input distribution, we pick
\begin{align*}
&U_i = \ol{V}_i := \wtild{V}_i - \wtild{V}_{ie},\quad U_i', U_i: \text{i.i.d.}\\
&X_{ic} = X_{ie} + \mathrm{Ber}_{1/2}[\mathrm{supp}V_i]\\
&X_i = X_{ic} + \mathrm{Ber}_{1/2}[\mathrm{supp}X_i\setminus \mathrm{supp}V_i]
\end{align*}
and $x_{ie}(u_1,u_2)$ is a linear map such that the random linear combinations of all levels of $u_1$ and $u_2$ is put uniformly at random on all the levels of $\mathrm{supp}X_i$. The support $\mathrm{supp}X_i$ and $\mathrm{supp}V_i$ denote the levels of $X_i$ and $V_i$ respectively, and $\mathrm{Ber}_{1/2}(\msf{A})$ denotes a random vector with i.i.d. $\mathrm{Ber}(1/2)$ random variables on all the levels of $\msf{A}$.

With the above choice, we have
\begin{align*}
r_1 &\ge I\lp U_1; \ol{V}_1\rp = H\lp \ol{V}_1\rp = p_2 n_{21}\\
r_2 &\ge I\lp U_2; \ol{V}_2\rp = H\lp \ol{V}_2\rp = p_1 n_{12}
\end{align*}
We choose $r_1 = p_2n_{21}$ and $r_2 = p_1n_{12}$, and obtain the following rate region after eliminating $r_1$, $r_2$ and redundant terms (see Appendix~\ref{app_RateEvaluate} for details):
\begin{align*}
R_{1p} &\le \msf{p_1} := (n_{11}-n_{21})^+\\
R_{2c} + R_{1p} &\le  \msf{s_1}:= \max\lbp (n_{11}-n_{21})^+, n_{12}\rbp\\
& \qquad\quad + p_1\min\lbp (n_{11}-n_{21})^+, n_{12}\rbp \\
R_1 &\le \msf{t_1}:= n_{11} + p_2(n_{21}-n_{11})^+ \\
R_{2c} + R_1 &\le \msf{n_1}:= \max(n_{11},n_{12}) \\
R_{2p} &\le \msf{p_2} := (n_{22}-n_{12})^+\\
R_{1c} + R_{2p} &\le  \msf{s_2}:= \max\lbp (n_{22}-n_{12})^+, n_{21}\rbp\\
& \qquad\quad + p_2\min\lbp (n_{22}-n_{12})^+, n_{21}\rbp \\
R_2 &\le \msf{t_2}:= n_{22} + p_1(n_{12}-n_{22})^+ \\
R_{1c} + R_2 &\le \msf{n_2}:= \max(n_{22},n_{21})
\end{align*}

After Fourier Motzkin elimination, $(R_1,R_2)$ satisfying
\begin{align*}
R_1 &\le \min\lbp \msf{t_1}, \msf{n_1}, \msf{p_1}+\msf{s_2}\rbp\\
R_2 &\le \min\lbp \msf{t_2}, \msf{n_2}, \msf{p_2}+\msf{s_1}\rbp\\
R_1+R_2 &\le \min\lbp \msf{p_1}+\msf{n_2}, \msf{p_2}+\msf{n_1}\rbp\\
R_1+R_2 &\le \msf{s_1}+\msf{s_2}\\
2R_1+R_2 &\le \msf{p_1}+\msf{n_1}+\msf{s_2}\\
R_1+2R_2 &\le \msf{p_2}+\msf{n_2}+\msf{s_1}
\end{align*}
is achievable, which coincides with the capacity region \eqref{eq_R1Bd} -- \eqref{eq_R12R2Bd} except the terms $\msf{p_1}+\msf{s_2}$ and $\msf{p_2}+\msf{s_1}$ in the individual rate constraints. We complete the achievability proof by the following fact.
\begin{fact}\label{fact_rate}
$\msf{t_1}\le \msf{p_1}+\msf{s_2},\ \msf{t_2}\le \msf{p_2}+\msf{s_1}$.
\end{fact}
\begin{proof}
See Appendix~\ref{app_Pf_Fact_rate}.
\end{proof}

\section{Converse Proof}\label{sec_Converse}
The converse proof is a novel modification of those in the perfect feedback case \cite{SuhTse_11} and the rate-limited feedback case \cite{VahidSuh_12}. Due to space constraints, below we outline the main proof and leave the details of the four useful facts to Appendix~\ref{app_Pf_Fact}.

\begin{fact}\label{fact_1}
For $(i,j) = (1,2), (2,1)$,
\begin{align*}
&X_i[t] \eqFunc \lp W_i, \wtild{V}_j^{t-1}, \ul{S}^{t-1}\rp \eqFunc \lp W_i, V_j^{t-1}, \ul{S}^{t-1}\rp.
\end{align*}
\end{fact}

\begin{fact}\label{fact_2}
For $(i,j) = (1,2), (2,1)$,
\begin{align*}
&H\lp Y_i^N|W_i,\ul{S}^N\rp = H\lp V_j^N,\wtild{V}_i^N|W_i,\ul{S}^N\rp.
\end{align*}
\end{fact}

\begin{fact}\label{fact_3}
For $(i,j) = (1,2), (2,1)$,
\begin{align*}
&I\lp W_i;V_j^N,\wtild{V}_i^N|\ul{S}^N\rp \le Np_jn_{ji}.
\end{align*}
\end{fact}

\begin{fact}\label{fact_4}
For $(i,j) = (1,2), (2,1)$,
\begin{align*}
&N^{-1}H\lp Y_i^N| V_i^N, \wtild{V}_j^N, \ul{S}^N\rp\\
&\le p_i (n_{ii}-n_{ji})^+ + (1-p_i)\max\lbp n_{ij}, (n_{ii}-n_{ji})^+\rbp,\\
&N^{-1}H\lp Y_i^N| V_j^N, \wtild{V}_i^N, \ul{S}^N\rp\\
&\le p_j(n_{ii}-n_{ji})^+ + (1-p_j)n_{ii},\\
\end{align*}
\end{fact}

\subsection{Bounds on $R_1$ and $R_2$}
We focus on the bounds on $R_1$. The first term in the minimization is a cut-set bound and the proof is in \cite{SuhTse_11}.

The second term is obtained as follows:
\begin{align*}
&N\lp R_1 - \epsilon_N\rp \le I\lp W_1; Y_1^N|\ul{S}^N\rp\\
&\le I\lp W_1; Y_1^N| V_2^N, \wtild{V}_1^N, \ul{S}^N\rp + I\lp W_1; V_2^N, \wtild{V}_1^N | \ul{S}^N\rp\\
&\overset{\aaaa}{\le} H\lp Y_1^N| V_2^N, \wtild{V}_1^N, \ul{S}^N\rp + Np_2n_{21}\\
&\le N\lbp p_2(n_{11}-n_{21})^+ + (1-p_2)n_{11} + p_2n_{21} \rbp\\
&= N\lbp n_{11} + p_2(n_{21}-n_{11})^+ \rbp.
\end{align*}
Here $\epsilon_N\ra 0$ as $N\ra\infty$. (a) is due to Fact~\ref{fact_3}.
Hence, \eqref{eq_R1Bd} holds. Similarly, so does \eqref{eq_R2Bd}.

\subsection{Bounds on $R_1+R_2$}
The first bound \eqref{eq_R1R2Bd1} is the bound when feedback is perfect, and the proof can be found in \cite{SuhTse_11}. 
The second bound \eqref{eq_R1R2Bd3} is non-trivial and is proved as follows. If $(R_1,R_2)$ is achievable,
\begin{align*}
&N\lp R_1+R_2 - \epsilon_N\rp
\le I\lp W_1;Y_1^N | \ul{S}^N\rp + I\lp W_2;Y_2^N | \ul{S}^N\rp \\
&\overset{\aaaa}{=} H\lp Y_1^N | \ul{S}^N\rp - H\lp V_1^N, \wtild{V}_2^N | W_2, \ul{S}^N \rp + H\lp Y_2^N | \ul{S}^N\rp\\
&\quad - H\lp V_2^N, \wtild{V}_1^N | W_1, \ul{S}^N \rp\\
&\le H\lp Y_1^N | V_1^N, \wtild{V}_2^N, \ul{S}^N\rp + I\lp W_2; V_1^N, \wtild{V}_2^N | \ul{S}^N \rp\\
&\quad + H\lp Y_2^N | V_2^N, \wtild{V}_1^N, \ul{S}^N\rp + I\lp W_1; V_2^N, \wtild{V}_1^N | \ul{S}^N \rp\\
&\overset{\bbbb}{\le} N\lbp\text{Right-Hand Side of \eqref{eq_R1R2Bd3}}\rbp.
\end{align*}
Here $\epsilon_N\ra 0$ as $N\ra\infty$. 
(a) is due to Fact~\ref{fact_2}. 
(b) is due to Fact~\ref{fact_3} and \ref{fact_4}. Hence \eqref{eq_R1R2Bd3} holds.

\subsection{Bounds on $2R_1+R_2$ and $R_1+2R_2$}
We focus on the bound in \eqref{eq_2R1R2Bd}. If $(R_1,R_2)$ is achievable,
\begin{align*}
&N\lp 2R_1 + R_2- \epsilon_N\rp\\
&\le I\lp W_1; Y_1^N|\ul{S}^N\rp + I\lp W_2; Y_2^N|\ul{S}^N\rp + I\lp W_1; Y_1^N|\ul{S}^N\rp\\
&\overset{\aaaa}{\le} H\lp Y_1^N|\ul{S}^N\rp - H\lp V_2^N, \wtild{V}_1^N|W_1, \ul{S}^N\rp + H\lp Y_2^N | \ul{S}^N\rp\\
&\quad - H\lp V_1^N,\wtild{V}_2^N|W_2,\ul{S}^N\rp + I\lp W_1; Y_1^N|W_2,\ul{S}^N\rp\\
&= H\lp Y_1^N|\ul{S}^N\rp - H\lp V_2^N, \wtild{V}_1^N| \ul{S}^N\rp + H\lp Y_2^N | \ul{S}^N\rp\\
&\quad + I\lp W_1; V_2^N, \wtild{V}_1^N| \ul{S}^N\rp - H\lp V_1^N,\wtild{V}_2^N|W_2,\ul{S}^N\rp\\
&\quad + H\lp Y_1^N|W_2,\ul{S}^N\rp\\
&\le H\lp Y_1^N|\ul{S}^N\rp + H\lp Y_2^N| V_2^N, \wtild{V}_1^N, \ul{S}^N\rp\\
&\quad + I\lp W_1; V_2^N, \wtild{V}_1^N| \ul{S}^N\rp + H\lp Y_1^N|V_1^N,\wtild{V}_2^N,W_2,\ul{S}^N\rp\\
&\overset{\bbbb}{\le} N\lbp\begin{array}{l}
\max(n_{11},n_{12}) + \max\lbp n_{21}, (n_{22}-n_{12})^+\rbp\\ + (n_{11}-n_{21})^+ + p_2\min\lbp n_{21}, (n_{22}-n_{12})^+\rbp
\end{array}\rbp.
\end{align*}
Here $\epsilon_N\ra 0$ as $N\ra\infty$. 
(a) is due to Fact~\ref{fact_2}. (b) is due to Fact~\ref{fact_3} and \ref{fact_4}, and 
\begin{align*}
&H\lp Y_1^N|V_1^N,\wtild{V}_2^N,W_2,\ul{S}^N\rp\\
&= H\lp Y_1^N|V_1^N,\wtild{V}_2^N,W_2,\ul{S}^N,X_2^N\rp\\
&= H\lp \mb{H}_{11}X_1^N|V_1^N,\wtild{V}_2^N,W_2,\ul{S}^N,X_2^N\rp
\le H\lp \mb{H}_{11}X_1^N|V_1^N\rp\\
&\le N(n_{11}-n_{21})^+ 
\end{align*}
due to Fact~\ref{fact_1}.
Hence, \eqref{eq_2R1R2Bd} holds.

\begin{remark}
In the above proof, we do not make use of the assumption that the realization of the feedback states $(S_1^N, S_2^N)$ is known to the transmitters causally. In other words, even if we allow some genie to provide the whole state sequence $(S_1^N, S_2^N)$ to all four terminals beforehand, the capacity region remains the same.
\end{remark}

\section{Acknowledgements}\label{sec_Ack}
The work of C. Karakus and S. Diggavi was supported in part by NSF award 1136174 and MURI award AFOSR FA9550-09-064. The work of I.-H. Wang was supported by EU project CONECT FP7-ICT-2009-257616.

%
\bibliographystyle{ieeetr}
\bibliography{Ref}

\appendices
\section{Achievable Rate Evaluation}\label{app_RateEvaluate}
Let us evaluate the rate constraints corresponding to the decoding at Rx1:
\begin{align*}
R_{1p} &\le I\lp X_1;Y_1, U_1, U_2 \mid X_{1c},X_{2c},U_1', U_2'\rp\\
&= H\lp Y_1, U_1, U_2 \mid X_{1c},X_{2c},U_1', U_2'\rp \\
&= H\lp X_1 \mid X_{1c}, X_{1e}, U_1', U_2'\rp\\
&= (n_{11}-n_{21})^+\\
R_{2c} + R_{1p} &\le I\lp X_{2c}, X_1;Y_1, U_1, U_2 \mid X_{1c},U_1', U_2'\rp\\
&= H\lp Y_1, U_1, U_2 \mid X_{1c},U_1', U_2'\rp\\
&= H\lp Y_1, U_2 \mid X_{1c}, X_{1e}, X_{2e}, U_1', U_2'\rp\\
&= \max\lbp (n_{11}-n_{21})^+, n_{12}\rbp\\
&\quad + p_1\min\lbp (n_{11}-n_{21})^+, n_{12}\rbp\\
R_1 &\le I\lp X_1;Y_1, U_1, U_2 \mid X_{2c},U_1', U_2'\rp\\
&= H\lp Y_1, U_1, U_2 \mid X_{2c},U_1', U_2'\rp\\
&= H\lp X_1, U_1 \mid X_{2c}, X_{1e}, X_{2e}, U_1', U_2'\rp\\
&= n_{11} + p_2(n_{21}-n_{11})^+\\
R_{2c} + R_1 &\le I\lp X_{2c}, X_1;Y_1, U_1, U_2 \mid U_1', U_2'\rp\\
&= H\lp Y_1, U_1, U_2 \mid U_1', U_2'\rp\\
&= H\lp Y_1, U_1, U_2 \mid X_{1e}, X_{2e}, U_1', U_2'\rp\\
&= \lbp\begin{array}{l}
q_{00}\max(n_{11},n_{12})\\ + q_{01}\max(n_{11},n_{12}+n_{21})\\ + q_{10}(n_{11}+n_{12})\\ + q_{11}(\max(n_{11},n_{21})+n_{12})
\end{array}\rbp\\
r_1 + r_2 + R_{2c} + R_1 &\le I\lp U_1', U_2', X_{2c}, X_1;Y_1, U_1, U_2 \rp\\
&= H\lp Y_1, U_1, U_2\rp \\
&= \max(n_{11},n_{12}) + p_2n_{21} + p_1n_{12}
\end{align*}

Plug in $r_1=p_2n_{21}, r_2=p_1n_{12}$ we obtain
\begin{align*}
R_{1p} &\le (n_{11}-n_{21})^+\\
R_{2c} + R_{1p} &\le \max\lbp (n_{11}-n_{21})^+, n_{12}\rbp\\
&\quad + p_1\min\lbp (n_{11}-n_{21})^+, n_{12}\rbp\\
R_1 &\le n_{11} + p_2(n_{21}-n_{11})^+\\
R_{2c} + R_1 &\le 
\lbp\begin{array}{l}
q_{00}\max(n_{11},n_{12})\\ + q_{01}\max(n_{11},n_{12}+n_{21})\\ + q_{10}(n_{11}+n_{12})\\ + q_{11}(\max(n_{11},n_{21})+n_{12})
\end{array}\rbp\\
R_{2c} + R_1 &\le \max(n_{11},n_{12})
\end{align*}

Below we show that the second bound on $R_{2c} + R_1$ always dominates the first term. First note that $q_{00}+q_{01}+q_{10}+q_{11}=1$. Hence, we separate right hand side of the second bound into four parts, and show each of them is not smaller than the corresponding ones in the first bound.
\begin{itemize}
\item $q_{00}$-term: $\max(n_{11},n_{12})$ = $\max(n_{11},n_{12})$.
\item $q_{01}$-term: $\max(n_{11},n_{12}+n_{21}) \ge \max(n_{11},n_{12})$.
\item $q_{10}$-term: $n_{11}+n_{12} \ge \max(n_{11},n_{12})$.
\item $q_{11}$-term: $(\max(n_{11},n_{21})+n_{12}) \ge \max(n_{11},n_{12})$.
\end{itemize}
Hence, we obtain the rate region.

\section{Proof of Fact~\ref{fact_rate}}\label{app_Pf_Fact_rate}
Note that both $\msf{t_1}$ and $\msf{p_1}+\msf{s_2}$ are affine in $p_2$. Plugging in $p_2=0$, we see that
\begin{align*}
&\msf{t_1} = n_{11}\\
&\le \msf{p_1}+\msf{s_2} = \max\lbp (n_{22}-n_{12})^++(n_{11}-n_{21})^+, n_{21}, n_{11}\rbp.
\end{align*}
Plugging in $p_2=1$, we see that
\begin{align*}
&\msf{t_1} = \max(n_{11},n_{21})\\
&\le \msf{p_1}+\msf{s_2} = (n_{11}-n_{21})^+ + n_{21} + (n_{22}-n_{12})^+ .
\end{align*} 
Hence, $\msf{t_1}\le \msf{p_1}+\msf{s_2}$ for all $p_2$. Similarly, $\msf{t_2}\le \msf{p_2}+\msf{s_1}$.

\section{Proof of Fact~\ref{fact_1} -- \ref{fact_4}}\label{app_Pf_Fact}
\subsection{Proof of Fact~\ref{fact_1}}
By definition, we have
\begin{align*}
X_1[t] &\eqFunc \lp W_1, \wtild{Y}_1^{t-1}, \ul{S}^{t-1}\rp \eqFunc \lp W_1, X_1^{t-1}, \wtild{V}_2^{t-1}, \ul{S}^{t-1}\rp\\
&\eqFunc \lp W_1, \wtild{V}_2^{t-1}, \ul{S}^{t-1}\rp \eqFunc \lp W_1, V_2^{t-1}, \ul{S}^{t-1}\rp.
\end{align*}
Similarly, the other functional relationship holds.

\subsection{Proof of Fact~\ref{fact_2}}
Let us focus on the first equality. 
\begin{align*}
&H\lp Y_1^N|W_1,\ul{S}^N\rp 
=\sum_{t=1}^N H\lp Y_1[t] | W_1, \ul{S}^N,Y_1^{t-1}\rp\\
&\overset{\aaaa}{=}\sum_{t=1}^N H\lp Y_1[t] | W_1, \ul{S}^N,Y_1^{t-1}, X_1^t\rp\\ 
&=\sum_{t=1}^N H\lp V_2[t] | W_1, \ul{S}^N,V_2^{t-1}, X_1^t\rp\\ 
&\overset{\bbbb}{=}\sum_{t=1}^N H\lp V_2[t] | W_1, \ul{S}^N,V_2^{t-1}\rp\\ 
&= H\lp V_2^N,\wtild{V}_1^N|W_1,\ul{S}^N\rp,
\end{align*}
where (a) is by definition, (b) is due to Fact~\ref{fact_1}.
The other holds similarly.

\subsection{Proof of Fact~\ref{fact_3}}
Below we shall prove the inequality for $(i,j)=(1,2)$. The second one follows in a similar way.
\begin{align*}
&I\lp W_1;V_2^N,\wtild{V}_1^N|\ul{S}^N\rp\\
&\overset{\aaaa}\le I\lp W_1;W_2,\wtild{V}_1^N|\ul{S}^N\rp
= I\lp W_1;\wtild{V}_1^N|\ul{S}^N,W_2\rp\\
&= H\lp\wtild{V}_1^N|\ul{S}^N,W_2\rp
\le H\lp\wtild{V}_1^N|S_2^N\rp\\
&= \E_{S_2^N}\lb H\lp (s_2V_1)^N \rp \big| S_2^N = s_2^N\rb\\
&\le \E_{S_2^N}\lb \sum_{t=1}^N H\lp s_2[t]V_1[t] \rp \Bigg| S_2^N = s_2^N\rb\\
&\le \E_{S_2^N}\lb N_1\lp s_2^N\rp n_{21} \Big| S_2^N = s_2^N\rb\\
&= Np_2n_{21}.
\end{align*}
Here $N_1(\cdot)$ denotes the number of $1$'s in the sequence. (a) follows because $V_2^N \eqFunc \lp W_2, \wtild{V}_1^{N}, \ul{S}^{N}\rp$.

\subsection{Proof of Fact~\ref{fact_4}}
We focus on $(i,j)=(1,2)$ below. For the first inequality,
\begin{align*}
&H\lp Y_1^N | V_1^N, \wtild{V}_2^N, \ul{S}^N\rp 
\le H\lp Y_1^N | V_1^N, \wtild{V}_2^N, S_1^N\rp\\
&= \E_{S_1^N}\lb H\lp Y_1^N | V_1^N, (s_1V_2)^N \rp \big| S_1^N = s_1^N\rb\\
&\le \E_{S_1^N}\lb \sum_{t=1}^N H\lp Y_1[t] | V_1[t], s_1[t]V_2[t]\rp \Bigg| S_1^N = s_1^N\rb\\
&\le \E_{S_1^N}\lb \left. \begin{array}{l}N_1\lp s_1^N\rp (n_{11}-n_{21})^+\\ + N_0\lp s_1^N\rp \max\lbp n_{12}, (n_{11}-n_{21})^+\rbp \end{array}\right | S_1^N = s_1^N\rb\\
&= Np_1 (n_{11}-n_{21})^+ + N(1-p_1)\max\lbp n_{12}, (n_{11}-n_{21})^+\rbp.
\end{align*}
Here $N_1(\cdot)$ and $N_0(\cdot)$ denote the number of $1$'s and $0$'s respectively in the sequence.

For the second inequality,
\begin{align*}
&H\lp Y_1^N | V_2^N, \wtild{V}_1^N, \ul{S}^N\rp 
\le H\lp Y_1^N | V_2^N, \wtild{V}_1^N, S_2^N\rp\\
&= \E_{S_2^N}\lb H\lp Y_1^N | V_2^N, (s_2V_1)^N \rp \big| S_2^N = s_2^N\rb\\
&\le \E_{S_2^N}\lb \sum_{t=1}^N H\lp Y_1[t] | V_2[t], s_2[t]V_1[t]\rp \Bigg| S_2^N = s_2^N\rb\\
&\le \E_{S_2^N}\lb N_1\lp s_2^N\rp (n_{11}-n_{21})^+ + N_0\lp s_2^N\rp n_{11} \Big| S_2^N = s_2^N\rb\\
&= Np_2 (n_{11}-n_{21})^+ + N(1-p_2)n_{11}.
\end{align*}

The other cases when $(i,j)=(2,1)$ follow similarly.

\end{document}